\documentclass[runningheads]{llncs}
\usepackage{amsmath, amssymb} 

\usepackage[T1]{fontenc}
\usepackage{tikz}
\usepackage{xcolor}
\usepackage{cite}
\usetikzlibrary{trees,positioning,decorations.pathreplacing,calc,patterns,fit}
\usepackage{color}
\usepackage{graphicx}
\usepackage{calc}
\usepackage{amsfonts}
\usepackage{xstring}
\usepackage[pdfencoding=auto]{hyperref}
\usepackage{cleveref}
\usepackage{caption}
\usepackage{subcaption}
\usepackage{ifthen}
\usepackage[misc]{ifsym}
\usepackage{pifont}
\usepackage{mathtools}
\usepackage{bbding} 

\hypersetup{
    colorlinks=true,
    linkcolor=blue,
    urlcolor=blue,
    linktoc=all,
    citecolor=blue
}

\definecolor{okabe1}{HTML}{000000}
\definecolor{okabe2}{HTML}{E69F00}
\definecolor{okabe3}{HTML}{56B4E9}
\definecolor{okabe4}{HTML}{009E73}
\definecolor{okabe5}{HTML}{F0E442}
\definecolor{okabe6}{HTML}{0072B2}
\definecolor{okabe7}{HTML}{D55E00}
\definecolor{okabe8}{HTML}{CC79A7}

\usepackage[]{todonotes} 

\newcommand{\Crefpart}[2]{%
	\nameCref{#1}~\hyperref[#2]{\labelcref*{#1}.\ref*{#2}}%
}
\renewcommand{\emph}[1]{\textit{\textbf{#1}}}
\let\epsilon\varepsilon

\DeclarePairedDelimiter\pars{\lparen}{\rparen}

\begin{document}
\title{Fast Geographic Routing in Fixed-Growth Graphs}
%
%
\author{Ofek Gila\inst{1}\orcidID{0009-0005-5931-771X} \and
Michael T. Goodrich\inst{1}\orcidID{0000-0002-8943-191X} \and
Abraham M. Illickan\inst{1}\orcidID{0009-0006-4410-7098} \and
Vinesh~Sridhar\inst{1}\orcidID{0009-0009-3549-9589}}
\authorrunning{Gila et al.}
%
\institute{University of California, Irvine, USA\\
\email{\{ogila,goodrich,vineshs1,aillicka\}@uci.edu}}

\setcounter{page}{0}
\maketitle              
\begin{abstract}
    In the 1960s, the social scientist Stanley Milgram performed his famous
    ``small-world'' experiments
    where he 
    found that people in the US who are far
    apart geographically are nevertheless connected by remarkably short chains
    of acquaintances.
    Since then, there has been considerable 
        work to design networks that
    accurately model the phenomenon that Milgram observed.
    One well-known approach
    was Barab{\'a}si and Albert's \emph{preferential attachment} model,
    which has small diameter yet lacks an algorithm that can efficiently find
    those short connections between nodes.
    Jon Kleinberg, in contrast,
    proposed a small-world graph formed from an $n \times n$ lattice that
    guarantees that greedy routing can navigate between any 
    two nodes in $\mathcal{O}(\log^2 n)$ time
    with high probability.
    Further work by Goodrich and Ozel and by Gila, Goodrich, and Ozel 
        present a hybrid technique that combines elements from these previous
    approaches to improve 
        greedy routing time to $\mathcal{O}(\log n)$ hops.
    These are
    important theoretical results, but we believe that their reliance on the
        square lattice limits their application in the real world.
    Indeed, a lattice
    enforces a fixed integral \emph{growth rate}, i.e. the rate of increase in
    the number of nodes at some distance $\ell$ reachable from some node $u$. 
    In
    this work, we generalize the model of Gila, Ozel, and Goodrich to any class
    of what we call \emph{fixed-growth} graphs of dimensionality $\alpha$, a subset of \textit{bounded-growth} graphs introduced in several
    prior papers.
    We prove tight bounds for greedy routing and diameter in these graphs, both
    in expectation and with high probability.
    We then apply our model to the U.S. road network to show that by modeling the
    network as a fixed-growth graph rather than as a lattice, we are able to
    improve greedy routing performance over all 50 states.
    We also show empirically that the optimal clustering exponent for the U.S.
    road network is much better modeled by the dimensionality of the network
    $\alpha$ than by the network's size, as was conjectured in a previous work.
\end{abstract}


\clearpage 
\section{Introduction}

In a series of experiments in the 1960s, Stanley Milgram explored what is now
known as the \emph{small-world phenomenon}
\cite{milgram1967small,travers1977experimental}.
Milgram asked random volunteers across the United States to send messages to
each other by mail.
If a participant did not know the designated final recipient, they were
instructed to send the message to a friend or family member they judged more
likely to know the intended recipient.
Remarkably, Milgram found that it took a median of only six hops for the
messages to reach their final destination despite the vast geographic distances
between participants. 

Milgram's groundbreaking work spurred a field of research that seeks to model
the small-world phenomenon, finding networks that capture the following
observations that he made.
First, the network must be tightly connected, linking arbitrary pairs of nodes
with a short chain of connections.
Second, the network must admit an efficient greedy algorithm that allows nodes
to find such short chains with just their local knowledge.
Third, links should be formed in a way that models how social connections are made
in the real world.

\subsection{Kleinberg's Model and Extensions}

The most well-known model of the small-world phenomenon was described by
Kleinberg \cite{kleinberg2000small}.
Similar to the earlier Watts and Strogatz model \cite{watts1998collective},
Kleinberg envisions the social network as a two-dimensional lattice, in which
individuals have local connections to the nodes adjacent to them in the lattice
as well as ``long-range'' contacts to randomly chosen other nodes.
Kleinberg showed that the graph enables a greedy routing protocol in expected
$\mathcal{O}(\log^2 n)$ hops (where the graph is an $n \times n$ lattice).
These greedy routing results were impressive, but the model itself was not
realistic---each node had the same degree. In contrast, we observe in 
real-world data that node degree generally follows a power law distribution. 
As such, the perfect lattice structure was not
correlated with real population density, and the $\mathcal{O}(\log^2{n})$ hop bound
of Kleinberg's model, which was 
later shown to be tight~\cite{martel}, is not small enough to explain Milgram's
results.

A recent paper by Goodrich and Ozel~\cite{goodrich2022modeling} sought to
address these limitations by creating a model where node degrees are drawn from
a power law distribution while the average degree remains constant.
They tested their model on road networks of all 50 U.S.~states as a 
better proxy for population density.
While they found that their model performed
well enough \textit{empirically} to explain Milgram's experiments, 
they did not provide any theoretical results.

Even more recently, Gila, Ozel, and Goodrich provided some theoretical
motivation for these empirical results, showing how a simpler model comprised of
two types of nodes, normal nodes and \emph{highway nodes}, could perform greedy
routing using only $\mathcal{O}(\log n)$ hops in
expectation~\cite{gila2023highway}.
Highway nodes contain long-range contacts to other highway nodes, 
forming a highway subgraph, $\mathcal{H}$. 
They have higher degrees than normal nodes, but were sparse enough 
such that the graph's average degree was still a constant.
They showed how results for their simpler model can be used to show good,
namely $\mathcal{O}(\log^{1 + \epsilon} n)$, greedy routing time for similar
graphs with a power law degree distribution.
However, they only proved their results for perfect two-dimensional lattices, and
showed that there was a gap between their expected greedy routing time to the
best possible lower bound.
Further, they did not provide high probability bounds on their greedy routing
time, nor did they provide bounds on the diameter.



\subsubsection{Other Small-World Models.}

There is a rich body of related work on small-world models, including
work in other topologies, preferential attachment models, and other classical
results.
We review this material in \Cref{sec:related-work}.

\subsection{Dimensionality in Graphs}

Most theoretical results proved for Kleinberg's model and its extensions assume
that the underlying network, before adding long-range contacts, is a lattice
$\mathcal{L}$~\cite{kleinberg2000small,martel,gila2023highway}.
Some work even considers lattices of different dimensions, showing that the
optimal \emph{clustering coefficient}, a
parameter that influences choice in long-range contacts, is
the dimension of the lattice~\cite{martel}.
One limitation of the recent results with power law degree distributions is that
since they assume the underlying network behaves like a two dimensional lattice,
they use a value of two for the clustering
coefficient~\cite{goodrich2022modeling,gila2023highway}.
The experimental paper by Goodrich and Ozel \cite{goodrich2022modeling} tested
different coefficients for two U.S.~states 
and found that the best value for each of
them was different, and both much closer to 1.5. 
They conjectured that the difference is
related to the size of the state.
Inspired by their results, we are interested in this
paper on a notion of ``dimensionality'' that
we could apply to more general graphs than lattices, 
including road networks, to give
more insight into their behavior as small worlds.

One of the earliest notions of graph dimensionality was developed by 
Erd{\"o}s,
Harary, and Tutte~\cite{Erdos_Harary_Tutte_1965}.
They define a dimension of a graph $\mathcal{G}$ geometrically: the smallest
$d$ such that $G$ can be embedded in $d$-dimensional Euclidean space such that
all edges are of unit distance.
Another notion popular in the study of networks is the \emph{doubling dimension}
\cite{gupta2003bounded,abraham2006routing,cole2006searching,konjevod2008dynamic,chan2016hierarchical,konjevod2007compact}.
A metric space is doubling if there is some constant $\lambda$ such that any
ball of size $r$ can be covered by the union of $\lambda$ balls of size $r/2$.
Abraham, Fiat, Goldberg, and Werneck introduced the \emph{highway dimension} of
a graph, providing a unified framework for understanding shortest-path
algorithms \cite{abrahamhighway}.
A graph has small highway dimension if, for some set of {\it access nodes} $S_r$
such that all shortest paths of length $> r$ include a vertex of $r$, every ball
of size $\mathcal{O}(r)$ contains some elements of $S_r$.
Work has also been done to translate the study of \emph{fractal dimensions} of
geometric objects, such as the Hausdorff or box counting dimensions, to studying
complex networks.
Fractal dimensions model graphs of ``non-integral''
growth. See \cite{rosenbergsurvey} for a survey of fractal dimensions and their
application to networks.

We are more interested, however, in the body of work that studies graphs with
some limitation on the rate of ``growth'' of the graph,
the rate of increase in the number of nodes in a ball of radius $\ell$ around any
node $u$, is bounded by some function of
$\ell$~\cite{stefanExpansion,johannesExpansion,kargerExpansion,ittai2005name,kuhn2005locality,krauthgamer2003intrinsic,BARTAL2005192,gfeller2007randomized}.
For example, some enforce a \emph{bounded-growth} requirement, that
the size of a ball of radius $\ell$ around any node contain at most
$\mathcal{O}(\ell^\alpha)$ nodes for some constant $\alpha$.
Our work considers a subset of such bounded-growth graphs with a slightly stricter requirement.
We consider the family of \emph{fixed-growth} graphs 
$\mathcal{F}$, where the number of
nodes in balls of radius $\ell$ around any node is both lower and upper bounded
by values in  $\Theta(\ell^\alpha)$ for reasonable values of $\ell$, where
$\alpha$ can be considered the dimensionality of the graph. 
Lattices of any dimension $\alpha$ are examples of fixed-growth graphs.

A paper by Duchon, Hanusse, Lebhar, and Schabanel extends Kleinberg's results to
more general graphs following a similar requirement~\cite{duchon2006could}.
In that paper, however, they were only able to show that greedy routing runs in
some polylogarithmic time, without determining the specific exponent, achieving
worse results than Kleinberg's let alone those of more recent papers.



\subsection{Our Results}

In this paper, we extend on the theoretical work 
of Gila, Ozel, and Goodrich~\cite{gila2023highway},
providing tight bounds on greedy routing time and diameter.
Further, in addition to greedy routing results \textit{in expectation}, we
provide results for greedy routing with high probability.
We prove all these results for a more general class of graphs, fixed-growth
graphs.

We motivate our use of the fixed-growth model with empirical results on U.S.
road networks, showing how by modeling U.S. road networks as fixed-growth graphs
we are able to pick clustering coefficients that perform much better than by
just assuming a lattice structure.
Further, we show how our notion of fixed-growth dimensionality is a much better
indicator of the optimal clustering coefficient than the size of the network, as
was conjectured in~\cite{goodrich2022modeling}.

We believe that by extending the results of Gila {\it et al.} to fixed-growth graphs, our work enables more accurate analyses of phenomena
modeled using small-world graphs such as disease spread
\cite{BARTHELEMY20111,warren2001firewallsdisorderpercolationepidemics,SANDER20031,LI2021111294,costaanalyzing2011},
decentralized peer-to-peer communication
\cite{dahliaviceroy2002,zhangusing2002,qindongsecure2022,shinsmall2011,li2005searching,singla2012jellyfish,hui2004small},
and gossip protocols \cite{kempe2004spatial,kempe2002protocols}.

\section{Definitions}

In general, small-world models take an existing undirected connected graph $G$,
such as a lattice $\mathcal{L}$, and augment it into a small-world graph
$\mathcal G$ by adding \textit{directed} long-range edges.
The original edges of $G$ are referred to as \textit{local contacts}, and the
new edges are referred to as \textit{long-range contacts}.
The probability
that a long-range contact added to node $u$ is another node $v$
is proportional to $d(u,v)^{-s}$, where $s$ is called the \textit{clustering
coefficient}, and $d(u,v)$ denotes the distance between $u$ and $v$ using only local
contacts in the original graph $G$.
A greater clustering coefficient biases long-range contacts to nodes that are
nearer to $u$.
A ``ball'' of radius $\ell$ centered around a node $u$, denoted as
$\mathcal B_\ell(u)$%
, is
the set of all nodes within distance $\ell$ of $u$.



We define a graph family $\mathcal{F}$ as having \emph{fixed-growth} (FG)
dimensionality $\alpha$ if, for all possible graph sizes $n$,
for all nodes $u \in \mathcal{F}(n)$, and for all radii
$1 \leq \ell \leq \Theta(\sqrt[\alpha]{n})$, $|\mathcal B_\ell(u)| =
\Theta(\ell^\alpha)$, where the ball of some radius $\ell \in
\Theta(\sqrt[\alpha]{n})$ encompasses the entire graph.\footnote{
Below, we loosely use the term ``graph'' to denote entire graph families.}
Note that the constants hidden in the $\Theta$ notation are defined over the
entire graph family $\mathcal{F}$ and do not depend on the graph size $n$.
Also observe that all fixed-growth graphs
have constant degree, since the number of nodes at distance 1 from any node,
$|\mathcal B_1(u)|$, is $\Theta(1)$.
Nodes in $\alpha$-dimensional lattices have degree $2 \alpha$, which is indeed
constant with respect to $\alpha$.
We compare the growth of balls in fixed-growth graphs to that of lattices in
\Cref{fig:balls-and-shells}.

Throughout the paper, we denote certain events as having high probability in
$x$, where $x$ is a variable such as $n$ or a function of a variable such as
$\log n$. An event succeeds with high probability in $x$ if it succeeds with
probability at least $1 - x^{-c}$ for some constant $c > 1$. 


\begin{figure}[t]
	\centering
	\includegraphics[height=.28\textheight]{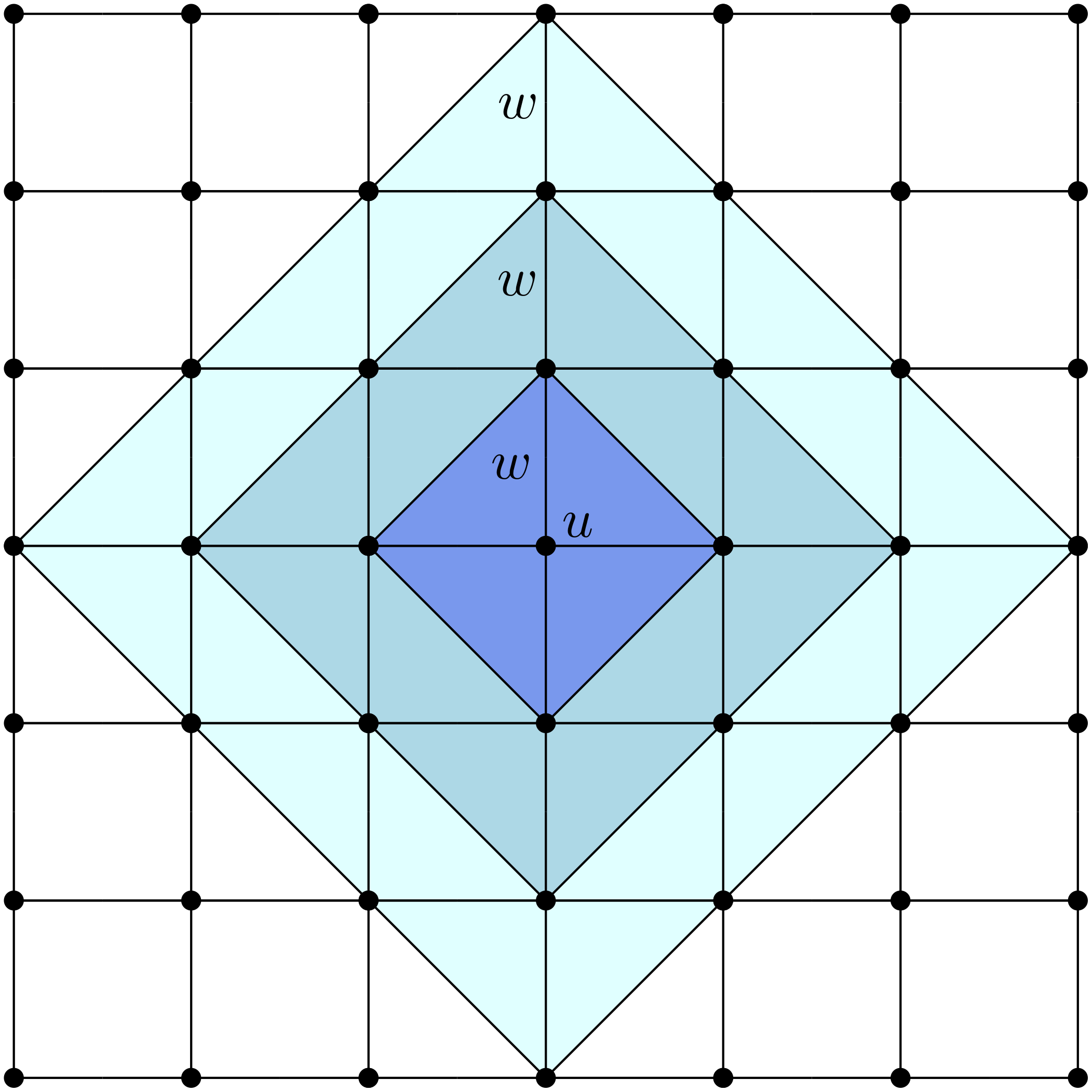}
	\hfill
	\includegraphics[height=.28\textheight]{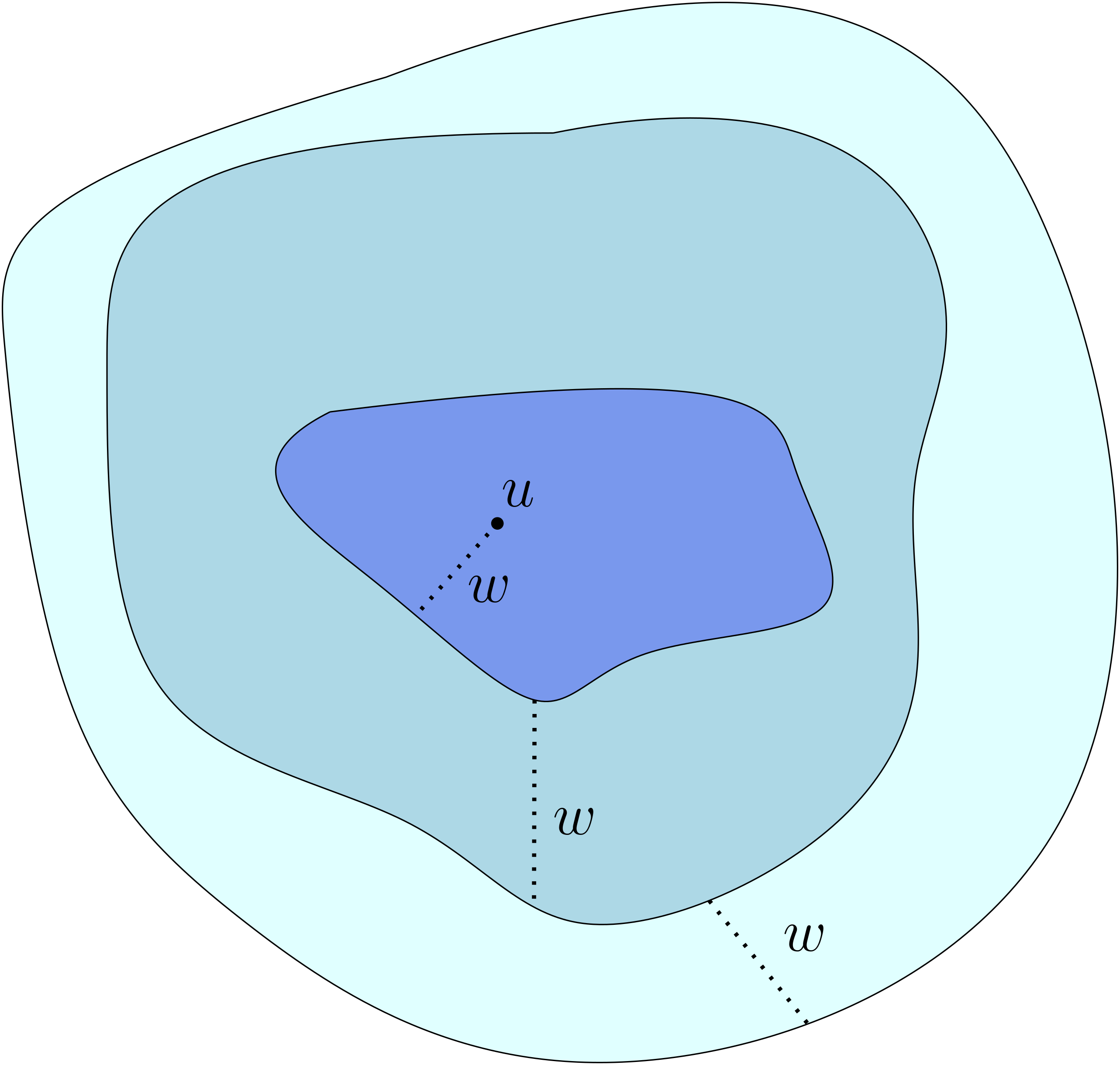}
	\caption{
		In these figures we show balls of various widths surrounding some
		central node $u$ in a lattice graph (left) and a fixed-growth graph
        (right).
		Note how balls grow far more predictably in lattices than they do in
		general fixed-growth graphs.
		We refer to the shaded regions in between two balls as shells, and the
		difference in radius between two balls as the shell width $w$.
        \label{fig:balls-and-shells}
    }
	\vspace*{-\medskipamount}
\end{figure}


\subsection{Randomized Highway Model}
The \emph{randomized highway model} introduced by 
Gila, Ozel, and Goodrich
takes a
2-dimensional lattice $\mathcal L$ and augments it into a small-world graph
$\mathcal K$ with some highway constant $k$~\cite{gila2023highway},
by adding long-range contacts with a clustering coefficient of 2.
However, unlike in previous models, not all nodes have long-range
contacts.
Rather,
each node
in $\mathcal L$ independently has $1/k$ chance of becoming a {\it highway node},
where only highway nodes have long-range contacts.
Because there are fewer highway nodes, each is able to have more long-range
contacts while maintaining a constant average degree.
Specifically,
to create the small-world graph $\mathcal K$, they add $q \times k$
long-range contacts to each highway node for some constant $q$.
As a result, the average degree of $\mathcal{K}$ is $q$.
Further, these
long-range contacts exclusively point to other highway nodes.

In our paper we extend the randomized highway model to a more general setting.
Consider an arbitrary fixed-growth graph $\mathcal F$ with dimensionality
$\alpha$.
We turn it into a randomized highway graph $\mathcal{G}$ by adding long-range
contacts with a clustering coefficient of $\alpha$ rather than 2.
We also allow each highway node to add $\Theta(k)$ long-range contacts rather
than strictly $q \times k$, where the constant bounds are the same across all
nodes.
\section{Our Results in the Fixed-Growth Model}

In this section, we present several results for highway graphs built from more
general underlying graphs.
As discovered by \cite{gila2023highway}, we obtain optimal results when $k =
\Theta(\log n)$, and use this as the canonical value for our key results.

\subsection{Preliminary Results}

In this section we go over results which are essential for 
our bounds on greedy routing and diameter.

\subsubsection{Balls.}
One of the key difficulties overcome in the original randomized Kleinberg
highway model was dealing with dependent probabilities.
For example, if there are only a few highway nodes close to the destination $t$, the probability that
a highway node along the greedy routing path has a long-range contact near $t$ will
be smaller for every highway node visited.
Gila, Ozel, and Goodrich solved this in two key ways---assuming all edges are
directed\footnote{Using undirected edges would be a strict improvement.}, and by
subdividing the graph into separate, disjoint diamonds (balls) \cite{gila2023highway}.
Each ball is large enough such that the number of highway nodes it contains can
be bounded with high probability.
Here we generalize these results for $\alpha$-dimensional fixed-growth graphs:
\begin{lemma} \label{lem:balls}
    Results for any randomized highway graph $\mathcal{G}$ with FG dimensionality $\alpha$ and highway constant $k$:
    \begin{enumerate}
        \item There exist balls of radius $\Omega(\sqrt[\alpha]{k \log{n}}) \leq
        \ell \leq \sqrt[\alpha]{n}$, centered around any of the $n$ nodes,
        containing $\Theta(\ell^\alpha / k)$ highway nodes with high probability in
        $n$. \label{lem:balls:always}


        \item With high probability in
        $\log{n}$, the balls of radius $\Theta(\sqrt[\alpha]{k
        \log{\log{n}}})$ centered around any $\Theta(\log^2{n})$ nodes each contain
        $\Theta(\log{\log{n}})$ highway nodes. \label{lem:balls:log}
        
        \item Any arbitrary ball of radius $\Theta(\sqrt[\alpha]{k})$ contains
        no highway nodes with constant probability.
        This result is not a high probability bound, and is only independent for
        balls centered around nodes with distance at least $c \sqrt[\alpha]{k}$,
        for some constant $c$ hidden in the big-$\mathcal{O}$ notation.
        \label{lem:balls:constant}
    \end{enumerate}
\end{lemma}

\begin{proof}
    Since $\mathcal{G}$ is an $\alpha$-dimensional fixed-growth graph, we know
    that there are $\Theta(\ell^\alpha)$ nodes in any ball of radius $\ell \in
    \mathcal{O}(\sqrt[\alpha]{n})$.
    There exists some constant $c_1$ such that, for any node $u$ in the graph,
    $|B_\ell(u)| \geq c_1 \ell^\alpha$.
    Letting $\ell = c_2 \sqrt[\alpha]{k \log{n}}$ for some constant $c_2$, there
    are at least $c_1 c_2^\alpha k \log{n}$ nodes in $B_\ell(u)$.
    Because the probability of any node being a highway node is $1/k$, the
    expected number of highway nodes in $B_\ell(u)$ is at least $c_1
    c_2^\alpha \log{n}$.
    Letting $X$ be the number of highway nodes in $B_\ell(u)$, we can apply the
    simplified Chernoff bounds of~\cite{DILLENCOURT2023106397} to bound the number
    of highway nodes in the ball:
    \begin{equation*}
        \Pr(X < (1 - \delta) \mu) < 2^{-\mu} \leq 2^{-c_1 c_2^\alpha \log{n}} = n^{-c_1 c_2^\alpha}
    \end{equation*}
    Note that this holds for $\delta > 0.91$.
    Applying a union bound over all $n$ nodes in the graph, we can bound the
    probability that any such ball contains fewer nodes by $n^{1 - c_1
    c_2^\alpha}$.
    The constant factor in the ball radius, $c_2$, can be chosen to be large
    enough to make this a high probability bound.
    For an upper bound, we apply the same reasoning, assuming that $|B_\ell(u)|
    \le c_3 \ell^\alpha$ for some constant $c_3 > c_1$.
    For our aforementioned balls, we expect to have at most $c_3 c_2^\alpha
    \log{n}$ highway nodes.
    Using a similar Chernoff bound from~\cite{DILLENCOURT2023106397}, for
    $\delta > 2.21$, we have:
    \begin{equation*}
        \Pr(X > (1 + \delta) \mu) < 2^{-\delta \mu} \leq 2^{-\delta c_3 c_2^\alpha \log{n}} = n^{-\delta c_3 c_2^\alpha}
    \end{equation*}

    In this case, the probability any such ball contains more nodes can be
    bounded by $n^{1 - \delta c_3 c_2^\alpha}$.
    $\delta$ can now be chosen to be sufficiently large such that this is a high
    probability bound.
    Combining the lower and upper bounds, we see that balls of radius $\ell =
    c_2 \sqrt[\alpha]{k \log{n}}$ contain $\Theta(\log{n}) = \Theta(\ell^\alpha
    / k)$ highway nodes with high probability in $n$.
    Since this holds for some $\ell$ in $\Theta(\sqrt[\alpha]{k \log{n}})$, it
    trivially holds for all larger balls 
    (observe that increasing $\ell$ would only make the exponent 
    derived from the Chernoff bounds smaller).
    This proves the first part of the lemma.
    The remaining parts of the lemma can be proven similarly.
    \qed
\end{proof}

\vspace*{-\bigskipamount}

\subsubsection{Shells.}

In the original Kleinberg model and its variants, the fact that the underlying
graph $\mathcal{L}$ was a perfect lattice was heavily utilized in the analyses 
\cite{kleinberg2000small,martel,gila2023highway}.
Unfortunately, we are unable to make such strict assumptions about the
underlying geometry of our fixed-growth graph $\mathcal{F}$.
One of the places this assumption breaks down is in the ability to neatly tile
independent balls the graph as is done in the nested lattice construction of
\cite{gila2023highway}.
In the fixed-growth setting, we have have no guarantees that, 
if we were to do this, no two balls would overlap.
This would violate their independence, making us unable to analyze how long-range
contacts are created.
We instead utilize the concept of a \emph{shell} when handling fixed-growth graphs.
Let $\mathcal{S}_b^{(w)}(u) = \mathcal{B}_{(b + 1)w_{b + 1}}(u) - \mathcal{B}_{bw_b}(u)$
be the shell at shell-distance $b$ and shell-width $w$ around node $u$.
Notice that this system of shells partitions the plane, 
even in the fixed-growth setting.
See \Cref{fig:balls-and-shells}.

\begin{lemma} \label{lem:shells}
    There exists a sequence of shell widths $w_b \in \Theta(\sqrt[\alpha]{k
    \log{n}})$, such that the number of highway nodes in a shell at
    shell-distance $b$ from any node is $\Theta(b^{\alpha - 1} \log{n})$ w.h.p.
    for any $1 \leq b \leq \Theta(\sqrt[\alpha]{n / (k \log{n})})$.
    For simplicity, we refer to any shell width $w_i$ as $w \in
    \Theta(\sqrt[\alpha]{k \log{n}})$, affecting at most constant factors.
\end{lemma}

\begin{proof}
    The number of nodes in this shell $|\mathcal{S}_b^{(w)}(u)|$ is the difference
    in the number of nodes between balls of radii $(b + 1)w_{b + 1}$ and $bw_b$, i.e.,
    $\Theta(((b + 1)w_{b + 1})^\alpha) - \Theta((bw_b)^\alpha)$.
    All $w_i$ are chosen from $\Theta(\sqrt[\alpha]{k \log{n}})$ such that
    the constants in this theta notation are the same, making the difference
    $\Theta(((b + 1)^\alpha - b^\alpha)w^\alpha)$.
    Through a simple polynomial expansion, we see that the total number of nodes in
    this ball is $\Theta(b^{\alpha - 1} w^\alpha)$ for any $b \ge 1$.
    Since each node has a $1/k$ chance of being a highway node, the expected number
    of highway nodes in this shell is $\Theta(b^{\alpha - 1} w^\alpha / k)$.
    As in the proof of \Cref{lem:balls}, in order to obtain a high probability
    bound over all $n$ nodes, we require at least $\Theta(\log{n})$ highway
    nodes in each shell.
    This is achieved by choosing $w = c_2 \sqrt[\alpha]{k \log{n}}$ for a
    sufficiently large constant $c_2$.
    \qed
\end{proof}

\subsubsection{Normalization Constant.}
A fundamental piece of information we need to know about any given
highway node $u$ is the probability that it is connected to another specific
highway node $v$.
In our fixed-growth model, in a graph $\mathcal{G}$ with dimensionality
$\alpha$, the probability is inversely proportional to the distance to the
power of $\alpha$, i.e., $\Pr(u \to v) \propto d(u, v)^{-\alpha}$.
To get the exact probability, we need to normalize this constant, i.e.,
$\Pr(u \to v) = d(u, v)^{-\alpha} / \sum_{h \in \mathcal H}(d(u, h)^{-\alpha})$, where
$\mathcal H$ is the subgraph of highway nodes in $\mathcal G$.
That denominator is the normalization constant, $z$. 
To lower bound
the probability, we must upper bound this constant.

\begin{lemma} \label{lem:z}
    Results for any randomized highway graph $\mathcal{G}$ with fixed-growth
    dimensionality $\alpha$ and highway constant $k$:
    \begin{enumerate}
        \item The normalization constant $z(u)$ is
        $\mathcal{O}(\frac{\log{n}}{k} + \log{\log{n}})$ for any highway node
        $u$ with high probability in $n$. \label{lem:z:always}

        \item The normalization constant $z(u)$ is
        $\mathcal{O}(\frac{\log{n}}{k} + \log{\log{\log{n}}})$ for
        $\mathcal{O}(\log^2{n})$ arbitrary highway nodes w.h.p. in $\log{n}$. \label{lem:z:log}

        \item The normalization constant $z(u)$ is
        $\mathcal{O}(\frac{\log{n}}{k})$ for an arbitrary highway node $u$
        with constant probability. \label{lem:z:constant}

        \item The normalization constant $z(u)$ is $\Omega(\frac{\log{n}}{k})$
        for any highway node $u$ with high probability in $n$.
        \label{lem:z:lower-always}
    \end{enumerate}
\end{lemma}

\begin{proof}
    Previous papers~\cite{kleinberg2000small,martel,gila2023highway} made
    extensive use of the rigid underlying geometry of the lattice to bound the
    normalization constant, specifically to determine the number of nodes at any
    distance from a given node.
    In our fixed-growth model, since we cannot make such strict assumptions, we
    instead sum over various shells with specific widths at different radii,
    using the results of \Cref{lem:shells}.
    First, let's consider the contribution to the normalization constant of an
    arbitrary highway node $u$ from all highway nodes at shell distance $b \geq
    1$, $z(u)_{b \geq 1}$. 
    Intuitively, we get the largest value of $z(u)_{b \geq 1}$ when all highway nodes
    in each shell are clustered as close to $u$ as they can be.
    \begin{align*}
        z(u)_{b \ge 1} &< \sum_{b = 1}^m{\frac{\text{max \# highway nodes in } \mathcal{S}_b^{(w)}(u)}{(\text{min distance to node in } \mathcal{S}_b^{(w)}(u))^\alpha}} \\
        &= \sum_{b = 1}^m{\frac{\Theta(b^{\alpha - 1} \log{n})}{(bw)^\alpha}} = \sum_{b = 1}^m{\frac{\Theta(b^{\alpha - 1} \log{n})}{b^\alpha k \log{n}}} \\
        &= \Theta\pars*{\frac{1}{k}} \sum_{b = 1}^m{\frac{1}{b}} = \Theta\pars*{\frac{\log{n}}{k}}
    \end{align*}
    This contribution is added to all the normalization constants w.h.p.
    A small adjustment to the above proofs provides results for the lower bound
    in part 4 of this lemma.
    Now let us consider the contribution of highway nodes within distance $w$ of
    $u$.
    From \Crefpart{lem:balls}{lem:balls:always} we know that there are at most
    $\Theta(\log{n})$ highway nodes within distance $w$ of $u$.
    By applying a derivative and as shown in~\cite{stefanExpansion}, there can be at
    most $\Theta(j^{\alpha - 1})$ nodes at distance $j$ from $u$.
    For our worst case analysis, we assume that all $\Theta(\log{n})$ nodes are
    packed tightly around $u$:
    \begin{equation*}
        z(u)_{b = 0} \leq \sum_{j = 1}^{\Theta(\sqrt[\alpha]{\log{n}})}{\frac{\Theta(j^{\alpha - 1})}{j^\alpha}} = \Theta(\log{\log{n}})
    \end{equation*}
    This concludes the proof for the first part of the lemma.
    For a slightly tighter bound, we limit how many nodes are packed around $u$
    using \Crefpart{lem:balls}{lem:balls:log}, where we learned that there are
    only $\Theta(\log{\log{n}})$ highway nodes within some inner distance
    $\Theta(\sqrt[\alpha]{k \log{\log{n}}})$ of $u$:
    \begin{equation*}
        z(u)_{b = 0, \text{ inner}} \leq \sum_{j = 1}^{\Theta(\log{\log{n}})}{\frac{\Theta(j^{\alpha - 1})}{j^\alpha}} = \Theta(\log{\log{\log{n}}})
    \end{equation*}
    Recall that we still have $\Theta(\log{n})$ unaccounted for highway nodes
    within distance $w$ of $u$.
    Again assuming the worst case possible, they are all at at the edge of the
    inner ball:
    \begin{equation*}
        z(u)_{b = 0, \text{ outer}} < \frac{\Theta(\log{n})}{\Theta(\sqrt[\alpha]{k \log{\log{n}}})^\alpha} = \Theta\pars*{\frac{\log{n}}{k \log{\log{n}}}}
    \end{equation*}
    This concludes the proof for the second part of the lemma.
    Finally, we prove the third part of this lemma where there are no highway
    nodes within some distance $\Theta(\sqrt[\alpha]{k})$ of $u$ which we know
    occurs with constant probability from
    \Crefpart{lem:balls}{lem:balls:constant}.
    Assuming all $\log{n}$ highway are all at distance
    $\Theta(\sqrt[\alpha]{k})$ from $u$, their total contribution would be at
    most $\Theta(\log{n} / k)$.
    \qed
\end{proof}

\subsubsection{Distance to the Highway.}

An important step of both the greedy routing algorithm and the diameter in
general is the ability to reach the highway.
From \Crefpart{lem:balls}{lem:balls:always}, we know that the maximum distance
between any arbitrary node and the nearest highway node is at most
$\mathcal{O}(\sqrt[\alpha]{k \log{n}})$.
In this section, we prove that this is a tight bound.
\begin{lemma} \label{lem:distance-to-highway}
    In any randomized highway graph $\mathcal{G}$ of FG dimensionality $\alpha$ with $k > 1 + \epsilon$, the maximum
    distance between any node and the nearest highway node is
    $\Theta(\sqrt[\alpha]{k \log{n}})$ w.h.p.
\end{lemma}
\begin{proof}
    To prove this tight bound, we will prove that for any ball of radius $\ell \in
    \mathcal{O}(\sqrt[\alpha]{k \log{n}})$, for sufficiently large $n$, there will exist
    some ball of radius $\ell$ that contains no highway nodes.
    The probability of any node being a highway node is $1/k$, and the
    probability it is not a highway node is $1 - 1/k$.
    Due to the definition of the fixed-growth model, there exists some constant
    $c_1$ such that for any node $u$ in the graph, $|B_\ell(u)| \leq c_1
    \ell^\alpha$.
    For any ball of radius $\ell$, the probability that it contains no highway
    nodes is at least $(1 - 1/k)^{c_1 \ell^\alpha}$.
    Due to the fact that for any $x < 1$, $1 - x \geq e^{-\frac{x}{1 - x}}$,
    this probability is at least $e^{-\frac{c_1 \ell^\alpha}{k - 1}}$, and
    consequently the probability that there is at least one highway node in the
    ball is at most $1 - e^{-\frac{c_1 \ell^\alpha}{k - 1}}$.
    Since $1 - x \leq e^{-x}$, this probability is at most $e^{-e^{-\frac{c_1
    \ell^\alpha}{k - 1}}}$.
    From \Cref{lem:indep-balls}, we know that there are $\Theta(n /
    \ell^\alpha)$ independent balls of radius $\ell$, which is at least $c_2 n /
    \ell^\alpha$ for some constant $c_2$.
    The probability that all balls contain at least one highway node is
    at most $1 - e^{-c_2 \frac{n}{\ell^\alpha} e^{-\frac{c_1
    \ell^\alpha}{k - 1}}}$ by a union bound.

    To show that at least one of these balls will be empty, we must show that
    this probability will be 0 for sufficiently large $n$, i.e., that
    $\lim_{n \to \infty} e^{-c_2 \frac{n}{\ell^\alpha} e^{-\frac{c_1
    \ell^\alpha}{k - 1}}} = 0$.
    We can see that this limit is 0 iff $\lim_{n \to \infty} c_2
    \frac{n}{\ell^\alpha} e^{-\frac{c_1 \ell^\alpha}{k - 1}} = \infty$.
    Taking the logarithm of both sides, we see that this is equivalent to
    $\lim_{n \to \infty} \log{c_2} + \log{n} - \alpha \log{\ell} - \frac{c_1'
    \ell^\alpha}{k - 1} = \infty$, where $c_1' = \frac{c_1}{\ln{2}}$.
    This limit showing that at least one ball will be empty will be satisfied if
    $\log{\ell} + \frac{c_1' \ell^\alpha}{k - 1} = \mathcal{O}(\log{n})$.
    Now let us finally assume that $\ell \in \mathcal{O}(\sqrt[\alpha]{k \log{n}})$.
    In this case, $\log{\ell} = \mathcal{O}(\log(k \log{n})) = \mathcal{O}(\log{k} +
    \log{\log{n}})$.
    This is $\mathcal{O}(\log{n})$ for any $k \in \mathcal{O}(n)$, which is indeed any
    reasonable value of $k$.
    For the second clause, we have that $\frac{c_1' \ell^\alpha}{k - 1} =
    c_1' \frac{k}{k - 1} \mathcal{O}(\log{n})$, which is $\mathcal{O}(\log{n})$ for any $k > 1 +
    \epsilon$.
    In other words, for sufficiently large $n$, there will exist at least one
    ball of radius $\ell \in \mathcal{O}(\sqrt[\alpha]{k \log{n}})$ that contains no
    highway nodes, so the maximum distance to a highway node must be at least
    $\Omega(\sqrt[\alpha]{k \log{n}})$.
    \qed
\end{proof}

\subsection{Greedy Routing in Expectation}
The key result of the small-world models is the ability to effectively
route between any two nodes of the graph using only local information.
In the original Kleinberg model the greedy routing time was shown to be
$\mathcal{O}(\log^2{n})$~\cite{kleinberg2000small}, a result later
proven tight~\cite{martel}.
Later papers have shown that highway variants achieve better results, either
empirically~\cite{goodrich2022modeling} or theoretically~\cite{gila2023highway}.
In the latter, the authors prove that greedy routing time can be reduced to 
$\mathcal{O}(\log{n})$ for two dimensional lattices by picking $k \in \Theta(\log{n})$.
In this section we generalize these results to fixed-growth graphs that are not
restricted to lattices, over any dimensionality $\alpha \ge 1$ including
non-integral values, and prove that these bounds are tight.
\begin{theorem} \label{thm:greedy-routing-exp} In any fixed-growth graph
    $\mathcal{G}$ with FG dimensionality $\alpha$ and highway constant $k \in
    \Theta(\log{n})$, greedy routing between two arbitrary nodes $s$ and $t$ succeeds in $\Theta(\log{n})$ hops with constant probability, if $d(s, t) =
    \Theta(\sqrt[\alpha]{n})$.
\end{theorem}

\begin{proof}
    We prove tight bounds for greedy routing considering three steps:
    \begin{enumerate}
        \item Reaching the highway using local contacts. \label{greedy-routing:step-1}

        \item Navigating along the highway using long-range contacts. \label{greedy-routing:step-2}
        
        \item Reaching the destination from the highway using local contacts. \label{greedy-routing:step-3}
    \end{enumerate}
    \crefname{enumi}{step}{steps}
    \Crefname{enumi}{Step}{Steps}
    %
    For \cref{greedy-routing:step-1}, by greedily walking using local contacts
    towards $t$ and because $k \in \Theta(\log n)$, we expect to 
    reach the highway in $\Theta(k)$ hops.
    Note that without any knowledge of the highway network, there is always a
    constant probability that we are unable to reach the highway in only $\Theta(k)$
    hops, serving as a trivial lower bound for this form of greedy routing.

    Once we reach the highway, we use long-range contacts to navigate the highway
    towards $t$ in \cref{greedy-routing:step-2}.
    From \Cref{lem:prob-halving}, we know that the probability that a long-range
    contact of a highway node $u$ improves its distance to $t$ by a
    factor of $c$ is at least proportional to $[(c + 1)^\alpha z(u)]^{-1}$ for some
    large enough constant $c$.
    We can improve our distance by a factor of $c$ at most $\log_c{d(s, t)} =
    \mathcal{O}(\log_c{\sqrt[\alpha]{n}}) = \mathcal{O}(\log{n})$ times
    before reaching $t$.
    If we assume that we are always able to use our tightest normalization
    constant bound from \Crefpart{lem:z}{lem:z:constant}, we expect $z(u)$ and
    consequently the probability of a factor $c$ improvement to be at most
    constant, and consequently expect to make $\mathcal{O}(\log{n})$ hops along
    the highway.
    This assumption is unreasonable, however, since it only holds with constant
    probability.
    Instead, we show in \Cref{lem:closer-contacts} that even when there is no factor of $c$
    improvement, we expect there to be long-range contacts with
    improvements in distance that allow us to visit enough fresh nodes to reach
    distance $\mathcal{O}(\sqrt[\alpha]{k})$ from $t$.
    Doing so incurs only a constant factor increase in the number of hops.
    Note that it is important to find a \textit{long-range contact} which
    improves our distance, as this guarantees that we are able to remain on the
    highway.
    Finally, we trivially complete \cref{greedy-routing:step-3} in
    $\mathcal{O}(\sqrt[\alpha]{k})$ time.
    \qed
\end{proof}

\subsection{Greedy Routing with High Probability}
In this section
, we go one step beyond several of the previous
analyses~\cite{kleinberg2000small,martel,gila2023highway} and 
prove high
probability bounds for greedy routing in fixed-growth graphs.
In order to do this, we make a small adjustment to the model, adding a constant
number of information to each non-highway node.
Specifically, each non-highway node knows which of its neighbors is the closest
to the highway, allowing it to reach the highway using the fewest number of
hops.
\begin{theorem} \label{thm:greedy-routing-hp}
    Let $\mathcal{G}$ be a randomized highway graph with FG dimensionality
    $\alpha$ and highway constant $k \in \Theta(\log{n})$.
    If $d(s, t) =
    \Theta(\sqrt[\alpha]{n})$, then
    greedy routing between any two nodes $s$ and $t$ succeeds in $\Theta(\log{n})$ hops with high
    probability in $\log{n}$ if $\alpha \geq 2$, and in
    $\Theta(\sqrt[\alpha]{\log^2{n}})$ hops if $\alpha \leq 2$.
\end{theorem}

\begin{proof}
    \crefname{enumi}{step}{steps}
    \Crefname{enumi}{Step}{Steps}
    The upper bound of this proof is very similar to the proof in expectation,
    with two main differences.
    First, in order to reach the highway for \cref{greedy-routing:step-1}, we
    make use of the information stored at each non-highway node, and due to the
    results of \Cref{lem:distance-to-highway}, we know that we can reach the
    highway in $\mathcal{O}(\sqrt[\alpha]{k \log{n}})$ hops with high
    probability.
    Second, while navigating the highway in \cref{greedy-routing:step-2}, we use
    the high probability result of \Cref{lem:closer-contacts} reaching only
    within distance $\mathcal{O}(\sqrt[\alpha]{k \log{n}})$ of $t$ rather than
    within distance $\mathcal{O}(\sqrt[\alpha]{k})$ as used in the expectation
    proof.

    For $\alpha \leq 2$, the $\Theta(\sqrt[\alpha]{k \log{n}})$ cost to reach
    the highway is $\Omega(\log{n})$ and so the bound is trivially tight.
    For larger values of $\alpha$, however, we must show that the time spent
    navigating the highway is greater than the time spent reaching the highway,
    i.e., that the time spent navigating the highway is $\Omega(\log{n})$.

    Consider a sequence of $f$ greedy hops along the highway, where $d_0$ and $d_f$
    are the distances from the first (resp. last) highway nodes on the path to
    $t$.
    In the best case, there always exists a long-range contact which improves
    our distance until we reach $d_f = \Theta(\sqrt[\alpha]{k \log{n}})$.
    Each hop improves our distance to $t$ by some factor.
    Let $x_i$ refer to the factor of improvement for the $i$-th hop (in greedy routing, it is always true that $x_i > 1$). Working
    backwards, we have that $d_f \times x_f \times x_{f - 1} \times \cdots
    \times x_0 = d_0$.
    Rearranging, we have that $\prod_i{x_i} = d_0 / d_f$.
    Taking the log of both sides, we have that $\sum_i{\log{x_i}} =
    \log{d_0} - \log{d_f}$.
    Since the first highway node is within distance $\mathcal{O}(\sqrt[\alpha]{k
    \log{n}})$ of $s$ and $d(s, t) = \Theta(\sqrt[\alpha]{n})$, then $\log{d_0}
    - \log{d_f} = \Theta(\log{n})$, and consequently $\sum_i{\log{x_i}} =
    \Theta(\log{n})$.
    If we can upper bound the expectation of $\log{x}$, $\mathbb{E}[\log{x}]$,
    then by linearity of expectation, we have that
    $\mathbb{E}[\sum_i{\log{x_i}}] = f \mathbb{E}[\log{x}]$.
    What is left to show is that $\mathbb{E}[\log{x}] = \mathcal{O}(1)$, such
    that the expected number of hops $\mu$ on the highway is $\Omega(\log{n})$.
    It is then possible to use a Chernoff bound to show that the probability of
    the number of hops being significantly less than $\mu$ is small.

    Now fix some constant $c_0 \geq 2$. 
    Due to the results of \Cref{cor:prob-halving-upper}, we know that the
    probability that some $x \geq c$ is at most proportional to $(c -
    1)^{-\alpha}$ for $c \geq c_0$ for some constant $c_0$.
    Note that we assume generously that $z(u)$ is always in $\Theta(1)$.
    When $x \leq c_0$, since $x$ is bounded by a constant, the contribution of
    $\log{x}$ to the sum is also at most constant.
    Now we consider $x > c_0$.
    Due to the definition of expectation, $\mathbb{E}_{x >
    c_0}[\log{x}] \leq \int_{c_0}^{\infty}{\log{c} \cdot \Pr(x = c) \text{ } dc}$, loosely
    summing as $c \to \infty$ even though $c$ can be more tightly bounded by
    $d_0 / d_f$.
    $\Pr(x = c)$ is certainly no greater than $\Pr(x \geq c) =
    \mathcal{O}(c^{-\alpha}) \leq \mathcal{O}(c^{-2})$ for all $\alpha \geq 2$.
    Finally, our results follow from the fact that
    $\int_{c_0}^{\infty}{\frac{\log{c}}{c^2} \text{ } dc} =
    \mathcal{O}(1)$.
    \qed
\end{proof}

\subsection{Diameter}

In this section we prove tight bounds on the diameter of randomized highway
graphs with fixed growth.
\begin{theorem} \label{thm:diameter}
    Let $\mathcal{G}$ be a randomized highway graph with FG dimensionality
    $\alpha$ and highway constant $k \in \Theta(\log{n})$.
    The diameter of $\mathcal{G}$ is $\Theta(\frac{\log{n}}{\log{\log{n}}})$ if
    $\alpha > 2$, and $\Theta(\sqrt[\alpha]{\log^2{n}})$ if $\alpha \leq 2$.
\end{theorem}

\begin{proof}
    The lower bound for the diameter is straightforward---from
    \Cref{lem:distance-to-highway}, there exists some node that is
    $\Omega(\sqrt[\alpha]{k \log{n}})$ hops from the highway, which is
    sufficient for $\alpha \leq 2$.
    For larger $\alpha$, we can reach the highway quickly, and the bottleneck
    rather becomes the time spent navigating the highway, $\mathcal{H}$.
    Let us consider the diameter of $\mathcal{H}$, starting from an arbitrary
    highway node $u$. 
    For a lower bound, we consider the most advantageous configuration.
    $u$ has $\Theta(k)$ long-range contacts, all to highway nodes, and
    $\Theta(1)$ local contacts, which we generously assume are
    also to highway nodes.
    Also, let us assume that, impossibly, all the $\Theta(k)$ contacts of
    all already-visited highway nodes are to new highway nodes at every step.
    Ignoring constant factors, there is at most a factor of $k$ increase in the
    number of nodes visited at each step.
    Since there are $\Theta(n/k)$ highway nodes, we can lower bound the amount
    of hops $f$ spent navigating the highway by:
    \begin{equation*}
        k^f \geq \frac{n}{k} \implies f \geq \frac{\log{n} - \log{k}}{\log{k}}
    \end{equation*}
    which is $\Omega(\frac{\log{n}}{\log{\log{n}}})$ when $k \in
    \Theta(\log{n})$.

    What remains to show is that it is indeed possible, with high
    probability, to navigate the highway $\mathcal{H}$ in
    $\mathcal{O}(\frac{\log{n}}{\log{\log{n}}})$ hops.
    In the standard approach for upper bounding the diameter~\cite{martel}, we
    build two large-enough sets $S$ and $T$ surrounding the source and
    destination nodes, respectively, and show that there exists an edge from $S$
    to $T$ with high probability.
    $S$ is a set of nodes that can reach be reached from $s$
    in at most
    $\ell$ hops (including taking long-range contacts), while $T$ can be thought of as its inverse, a set
    of nodes that can reach $t$ in at most $\ell$ hops.
    See \Cref{fig:diameter}.
    In our case, we are interested in the highway subgraph $\mathcal{H}$, and
    therefore construct sets $S$ and $T$ using only highway nodes and
    long-range contacts.

    Let us consider the growth of set $S$, referring to all the ``fresh'' nodes
    added at step $i$ as $\phi_i$, with $\phi_0 = \Theta(\log{n})$ consisting of the
    highway nodes within distance $\Theta(\sqrt[\alpha]{k \log{n}})$ of $s$.
    In the worst case, all the nodes in $\phi_i$ are as far away from other highway
    nodes as possible.
    Let $|S|$ be the final size of set $S$.
    If they were all bunched up together as close as possible, their maximum
    radius in the underlying graph $\mathcal{F}$ would be
    $\Theta(\sqrt[\alpha]{k |S|})$ as determined from
    \Crefpart{lem:balls}{lem:balls:always}.
    In order to reach $\Theta(n_\mathcal{H}^\eta)$ highway nodes, where
    $n_\mathcal{H} = \Theta(n / k)$ is the total number of highway nodes and
    $\eta < 1$ is some constant, we must only consider distances from $\phi$ up
    to $\Theta(\sqrt[\alpha]{k n_\mathcal{H}^\eta}) =
    \mathcal{O}(\sqrt[\alpha]{n^\theta})$ for some other constant $\theta < 1$.
    Using \Cref{lem:contacts-out} we state that the probability
    of a long-range contact $v$ of a highway node $u$ in $\phi_i$ to be fresh
    is at least proportional to $\frac{\log{n}}{k z(u)}$.
    Using the pessimistic normalization constant bound of
    \Crefpart{lem:z}{lem:z:always}, and the fact that $k \in \Omega(\log{n})$,
    the probability that $v$ is fresh is at least proportional to
    $\frac{\log{n}}{k \log{\log{n}}}$.
    The expected number of fresh nodes added at step $i$, $|\phi_{i + 1}|$,
    considering that each node in $\phi_i$ has $\Theta(k)$ long-range
    contacts, is at least proportional to $|\phi_i|
    \frac{\log{n}}{\log{\log{n}}}$.
    Recalling that the number of nodes in $|\phi_i|$ is always
    $\Omega(\log{n})$, the probability that there are asymptotically fewer fresh
    highway nodes is at most $n^{-\frac{\log{n}}{\log{\log{n}}}}$ using the
    simplified Chernoff bounds of~\cite{DILLENCOURT2023106397}.
    Equipped with a growth factor of $\frac{\log{n}}{\log{\log{n}}}$, we can
    apply similar reasoning to the diameter lower bound to state that we reach
    $n_\mathcal{H}^\eta$ total highway nodes in at most
    $\mathcal{O}(\frac{\log{n}}{\log{\log{n}} - \log{\log{\log{n}}}}) =
    \mathcal{O}(\frac{\log{n}}{\log{\log{n}}})$ hops.
    A similar argument along with the results of \Cref{cor:contacts-in} can
    grow set $T$ to similar magnitude.

    \begin{figure}[t]
        \centering
        \includegraphics[width=0.8\linewidth]{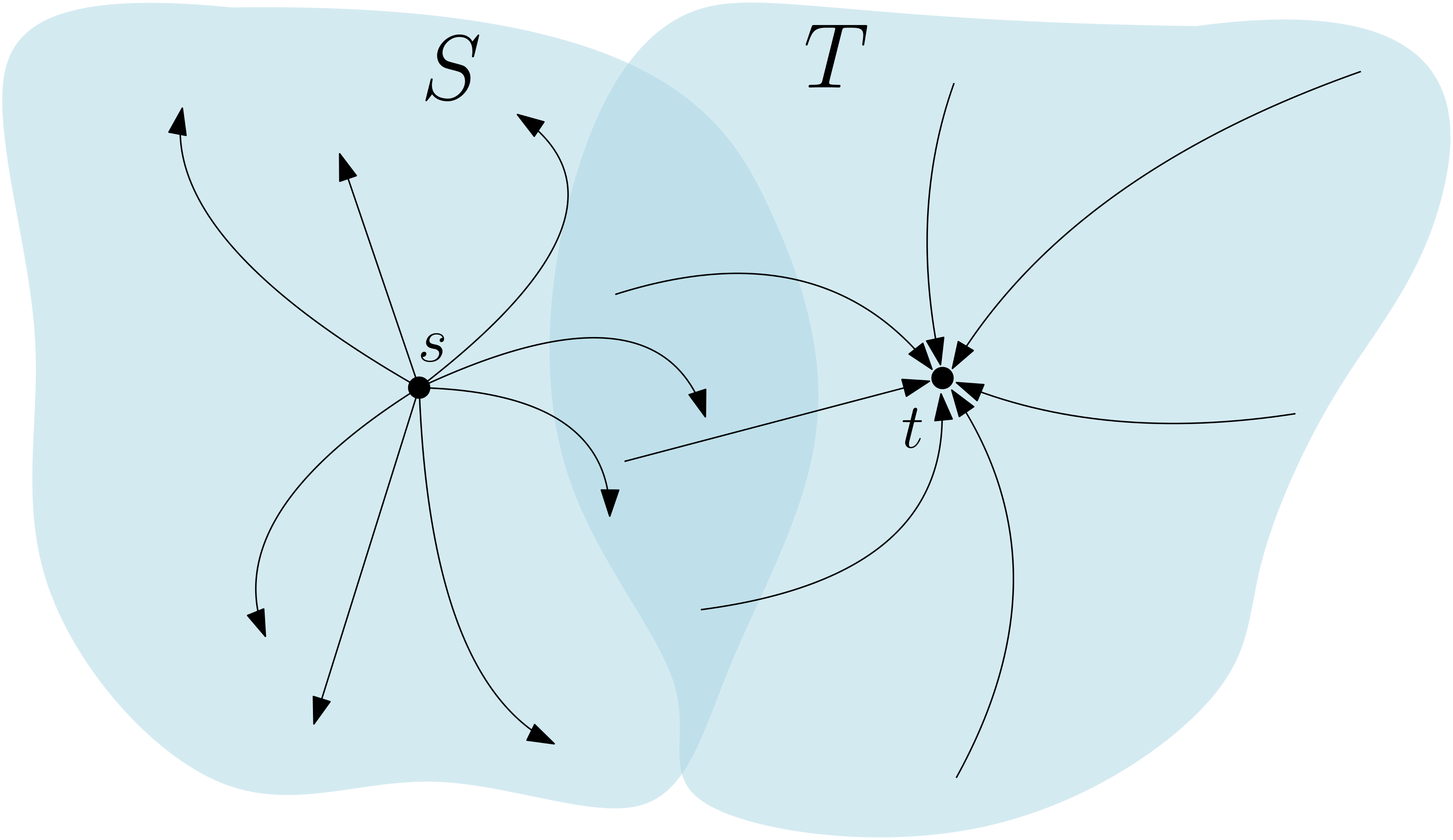}
        \caption{
            A visualization of the two sets $S$ and $T$.
            $S$ is the set of highway nodes that can be efficiently reached from
            the source $s$, while $T$ is the set of highway nodes that can
            efficiently reach the destination $t$.
            We show that these two sets intersect with high probability, and
            thus $s$ can route to $t$ efficiently.
        }
        \label{fig:diameter}
        \vspace*{-\medskipamount}
    \end{figure}

    If $S$ and $T$ intersect, a path exists between
    $\phi_0$ and $\psi_0$ in $\mathcal{H}$ in at most
    $\mathcal{O}(\frac{\log{n}}{\log{\log{n}}})$ hops, where $\psi_i$ is the
    set of fresh nodes added at step $i$ to $T$, completing our proof.
    We conclude by showing that if $S$ does not intersect with $T$ by
    step $f$, it almost surely will intersect with $T$ at the next step.
    It is easy to see that $|\phi_{f_1}|$ and $|\psi_{f_2}|$
    are both in $\Theta(n_\mathcal{H}^\eta)$.
    The maximum distance between any two nodes in $\mathcal{F}$ is
    $\Theta(\sqrt[\alpha]{n})$, and so the probability that an arbitrary
    long-range contact $v$ of a highway node $u$ in $\phi_f$ is a
    member of $\psi_f$ is at least proportional to $n_\mathcal{H}^\eta/(n
    z(u)) \geq n^{\eta - 1}/(k z(u))$.
    The probability that $v$ is not a member of $\psi_f$ is therefore at most
    $1 - n^{\eta - 1}/(k z(u)) \leq e^{-\frac{n^{\eta - 1}}{k z(u)}}$.
    The $\Theta(n_\mathcal{H}^\eta)$ members of $\phi_f$ each have $\Theta(k)$
    long-range contacts, and so the probability that none are in
    $\psi_f$ is at most $e^{-\frac{n^{2 \eta - 1}}{k z(u)}} \leq
    e^{-\frac{n^{2 \eta - 1}}{\log{n} \log{\log{n}}}}$ 
    for $k \in
    \Theta(\log{n})$.
    By picking $\eta > 0.75$, $n^{2 \eta - 1} \geq \sqrt{n}\log^3{n}$,
    the probability of failure is at most $e^{-\sqrt{n} \ln{n}} =
    n^{-\sqrt{n}}$, and of success is at least $1 -
    n^{-\sqrt{n}}$ which is certainly with high probability.
    \qed
\end{proof}

\section{Experiments}

One of the key assumptions made in previous papers is that the underlying
network is a lattice, and consequently for 2-dimensional data, the optimal
clustering exponent should be 2.
This assumption was even made in two recent papers that analyzed road networks
of the 50 U.S. states~\cite{goodrich2022modeling,gila2023highway} and DC, despite the
fact that these networks are certainly not laid out like perfect square grids.
Our work rejects the notion that the road networks are best modeled as a
lattice, and instead propose modeling them as fixed-growth graphs with some
dimensionality $\alpha$ that is not necessarily an integer.
In this section we present experiments to support this claim, showing how by
modeling the road networks as fixed-growth graphs, and by picking a clustering
coefficient as determined by the underlying dimensionality $\alpha$ of the graph
rather than assuming it is 2, we can achieve better routing performance for
\textit{every} U.S. state.
See \Cref{fig:greedy_routing_comp}.

It is important to note, however, that our definition of ``fixed-growth'' graphs
is not directly meaningful for finite graphs or graph families as the constants
hidden in the fixed-growth $\Theta$ can be picked to make any dimensionality
$\alpha$ work.
Nevertheless, we have found a heuristic to estimate
which dimensionality $\alpha$ is most appropriate for a given graph: 
minimizing the discrepancy in size between balls of the same radius. 
We also employ the same method as \cite{goodrich2022modeling} to 
empirically determine the best clustering coefficient%
.
We go into more detail in \Cref{sec:more-exp}.

\begin{figure}[t]
	\centering
	\includegraphics[width=\linewidth]{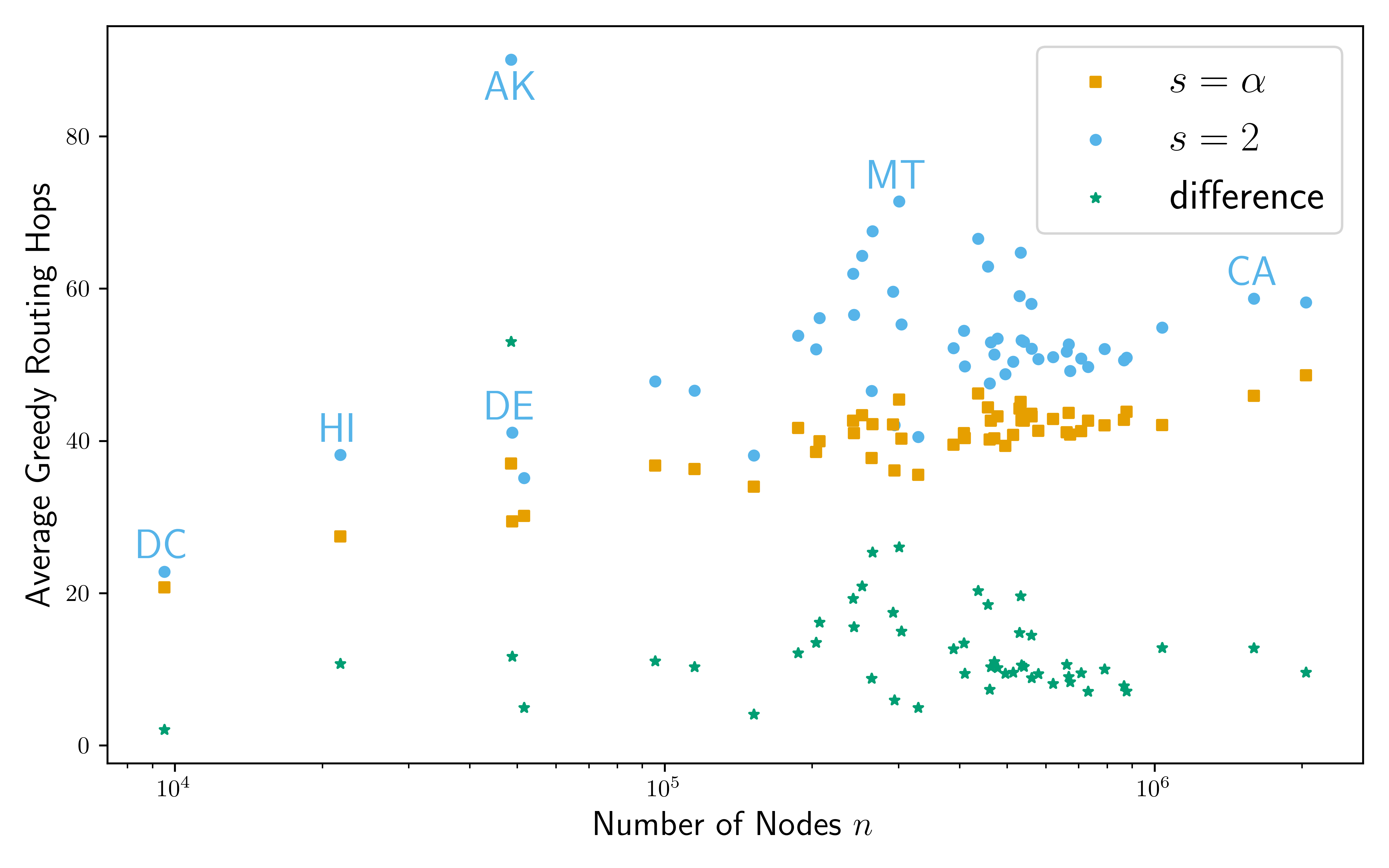}
	\vspace*{-\bigskipamount}
	\vspace*{-\medskipamount}
	\caption{
		Here we compare the average greedy routing performance of randomized
		highway graphs constructed from the road networks of all 50 U.S. states
		and DC using different clustering coefficients $s$.
		Specifically, we compare the performance for when $s$ is set to 2
		(blue), which is the value used by previous papers who assume that the
		road networks behave like a 2-dimensional lattice, to when $s$ is set to
		$\alpha$ (orange), the estimated dimensionality of the graph assuming
		that the road networks behave like fixed-growth graphs.
		Modeling them as fixed-growth graphs \textit{always} provides better
		greedy routing performance, as evident by their difference (green)
		always being positive.
		\label{fig:greedy_routing_comp}
	}
	\vspace*{-\medskipamount}
\end{figure}

\begin{figure}[t]
	\centering
	\includegraphics[width=\linewidth]{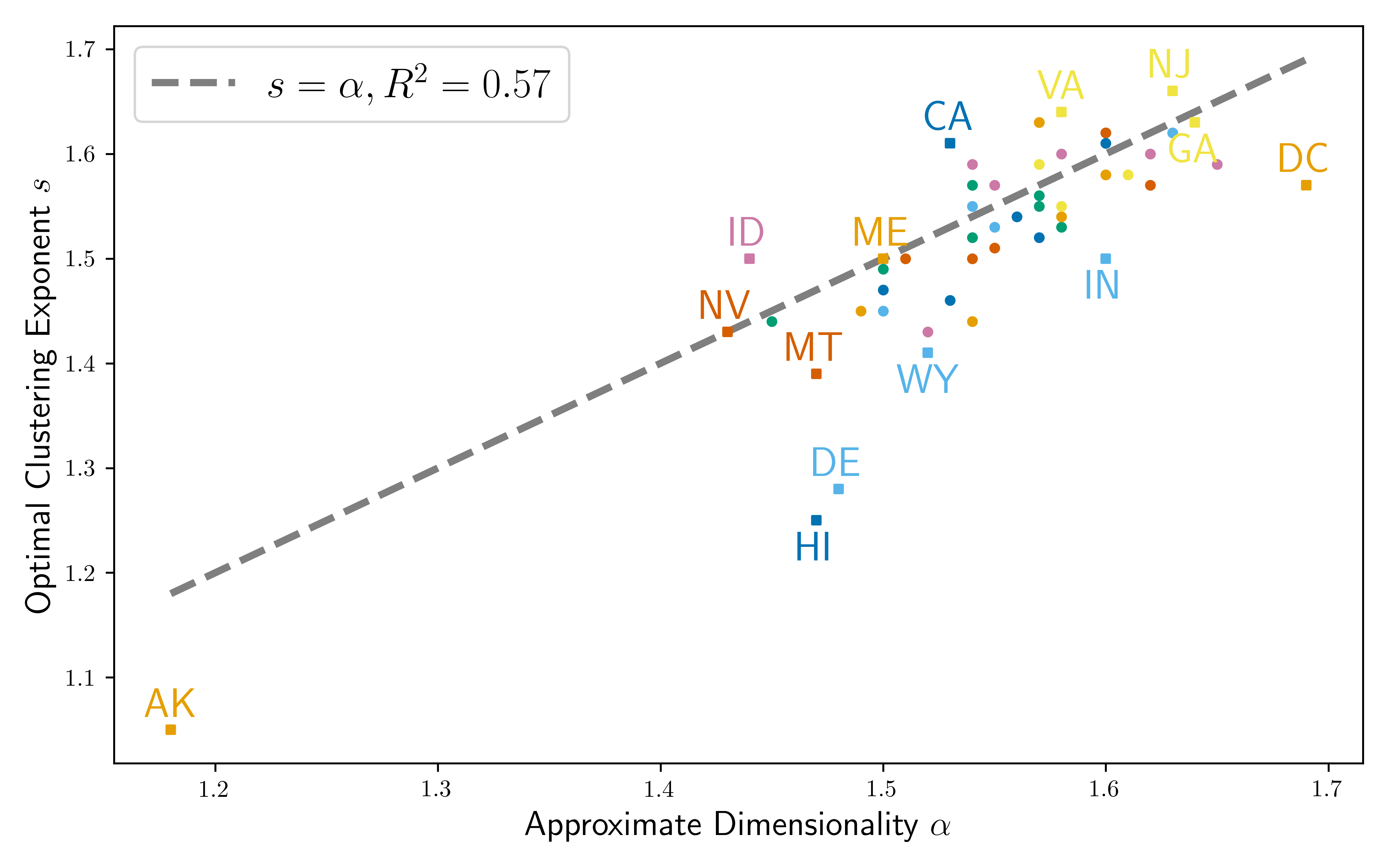}
	\vspace*{-\bigskipamount}
	\vspace*{-\medskipamount}
	\caption{
    Here we compare our estimated dimensionality $\alpha$, as used in \Cref{fig:greedy_routing_comp},
	to the empirically
    determined optimal clustering exponent $s$ for each state.
	There is a clear correlation between our values of $s$ and $\alpha$, and
    importantly, none of the optimal clustering coefficients are close to two,
    the assumed dimensionality in previous work
    \cite{gila2023highway,goodrich2022modeling,kleinberg2000small}.
    This suggests that our fixed-growth model is able to take advantage of the
    intrinsic growth rate of these graphs in a way that abstracting them to a
    lattice does not.
    \label{fig:clustering_exponent_vs_dimensionality}}
	\vspace*{-\medskipamount}
\end{figure}

We used the same U.S. road network data as the previous papers to provide an
estimate for the approximate dimensionality $\alpha$ and plotted it against the
empirically determined optimal clustering exponent $s$ in
\Cref{fig:clustering_exponent_vs_dimensionality}.
The first thing to note is that the estimated dimensionality $\alpha$ for U.S.
states is some fraction between 1 and 2.
Most previous small world analyses only apply to graphs with integer
dimensionality.
We observe that the correlation, while not perfect, is still clearly there.
Road networks that have much lower dimensionality such as Alaska
have a much smaller optimal clustering exponent than those with higher
dimensionality like that of Washington, DC.
And for all networks, the approximate dimensionality is a much better choice for
the clustering exponent than 2
as is further confirmed by \Cref{fig:greedy_routing_comp}. 

A previous paper empirically analyzing the U.S. road networks in
depth~\cite{goodrich2022modeling} conjectured that as the size of the state
increases, the optimal clustering exponent should increase tending towards 2,
based on a comparison of two U.S. states: Hawaii and California. 
In \Cref{sec:more-exp}, we present data from all 50 states showing
a much stronger correlation between the optimal clustering exponent and the
approximate dimensionality of the graph than to the state size.

\section{Future Work}

In this paper we have shown how to effectively model the highway networks of the
50 U.S. states as fixed-growth graphs with constant dimensionality $\alpha$.
We believe that this model can be extended to many other applications such as
for graphs that model disease spread, decentralized peer-to-peer communication
networks, and gossip protocols.
It would certainly be interesting to see how the model performs on these
applications.
We are also optimistic that this model can be used to provide bounds for various
highway constants $k$ other than $\Theta(\log{n})$.
Another interesting direction would be to consider the case where the highway
dimensionality $\alpha$ depends on either the ball radius $\ell$, the number
of nodes in the graph $n$, or is non-uniform across the graph.

\begin{credits}
	\subsubsection{\ackname} We thank Evrim Ozel for helpful discussions
	regarding the best ways to generalize lattices, importantly helping us
	settle on the fixed-growth model.
\end{credits}

\clearpage

\bibliographystyle{splncs04}
\bibliography{refs}

\clearpage
\appendix

\section{Related Prior Work} \label{sec:related-work}

Prior to Milgram's experiment, Pool and Kochen \cite{de1978contacts} as well as others \cite{kochensmallworld} found that if social connections were selected uniformly from the population, then the diameter of the graph is small with high probability. 
While interesting, these early results did not provide a compelling reason for the small-world phenomenon. 
In 1998, Watts and Strogatz suggested that individuals should have both ``local'' connections to adjacent nodes as well as ``long-range'' contacts like the ones studied prior \cite{watts1998collective}. 
This captured a sense of locality to social networks, and was able to model several natural and manmade phenomena. 
Kleinberg \cite{kleinberg2000small} generalized and improved their model, altering the criteria for choosing long-range contacts such that the graph enabled greedy routing between any two nodes to succeed with high probability in $\mathcal{O}(\log^2 n)$ time. 
Martel and Nguyen \cite{martel} later found that greedy routing in Kleinberg's model runs in $\Theta(\log^2 n)$ time yet the diameter of the graph is $\Theta(\log n)$, meaning that Kleinberg's algorithm fails to find the true shortest paths available in the graph by a $\log n$ factor. 
Other works have implemented Kleinberg's analyses in other topologies and types of graphs \cite{fraigniaudgreedy,kleinbergdynamics,barriereefficient,duchon2006could}. 

Another small-world model is the \emph{preferential attachment} model of Barabási and Albert \cite{barabasipreferential}. 
Here, nodes are added to the graph one at a time.
New nodes connect to randomly selected other nodes, with probability proportional to the other nodes' degree (a ``rich-get-richer'' process). Dommers, Hofstad, and Hooghiemstra \cite{dommers2010diameters} showed that the diameter of this model is low, yet unlike Kleinberg's model, no greedy algorithm exists to find such short paths between nodes. 
Since then, there have been efforts to combine the two models in order to get the best of both worlds.
Bringmann, Keusch, Lengler, Maus, and Molla propose such a model with an average $\mathcal{O}(\log\log n)$ greedy routing time, yet their model requires nodes of the network to be randomly distributed in some geometric space. 
Goodrich and Ozel take another approach, inspired by the U.S. highway system (which is ultimately what carried the messages shared in Milgram's experiment), called the \emph{neighborhood preferential attachment} model. 
Nodes are again added to the graph one at a time, but the probability of connecting to some other node $v$ is based both on the degree of $v$ as well as the distance between the new node and $v$. 
While their intuition was corroborated by their strong empiricial results as well as by Abraham, Fiat, Goldberg, and Werneck's work on the \emph{highway dimension} of graphs \cite{abrahamhighway}, Goodrich and Ozel were not able to prove any theoretical bounds for their model \cite{goodrich2022modeling}. 
In 2023, Gila, Ozel, and Goodrich \cite{gila2023highway} adapted the model of Goodrich and Ozel back to a lattice, their \emph{randomized highway} model, and were able to show that it expects to perform greedy routing in $\mathcal{O}(\log n)$ time. 

\clearpage

\section{Omitted Results and Proofs} \label{sec:misc-proofs}

\vspace*{-\medskipamount}

Here we provide additional lemmas and proofs omitted from
the main text.
\vspace*{-\medskipamount}

\subsubsection{The Number of Independent Balls.}
The following result tightly bounds the number of independent balls present
in the graph.
This result is important for lower bounding the distance to the highway, as is
done in the proof of \Cref{lem:distance-to-highway}.
\vspace*{-\medskipamount}
\vspace*{-\medskipamount}
\begin{lemma} \label{lem:indep-balls}
	A fixed-growth graph $\mathcal{G}$ with dimensionality $\alpha$ contains
	$\Theta\pars*{\frac{n}{\ell^\alpha}}$ independent balls of radius $\ell$.
\end{lemma}
\begin{proof}
	Consider the following procedure:
	\begin{enumerate}
		\item Consider all nodes in $\mathcal{G}$ as ``available''.
		
		\item Pick an arbitrary available node $u$ and mark all nodes in
		$\mathcal{B}_{2\ell}(u)$ as ``unavailable''.

		\item Add $u$ to the set of ball centers.
		
		\item Repeat steps 2 and 3 until all nodes are marked as unavailable.
	\end{enumerate}

	In step 2 we mark $\Theta(\ell^\alpha)$ nodes as unavailable,
	allowing our procedure to repeat $\Theta\pars*{\frac{n}{\ell^\alpha}}$
	times, picking a new ball center each time.
	Balls of radius $\ell$ centered around each center must be
	independent since they are at least $2\ell$ apart.
	\qed
\end{proof}

\vspace*{-\medskipamount}
\vspace*{-\medskipamount}

\subsubsection{The Probability of Distance Improving.} \label{sec:prob-halving}
In this section we lower bound the probability that a random long-range
contact $v$ of node $u$ improves its distance to another node $t$ by a constant
factor.
%
\begin{lemma} \label{lem:prob-halving}
	The probability that any of the long-range contacts of a highway node $u$
	improves its distance to the destination $t$ by a factor of $c$ is at least
	proportional to $[(c + 1)^\alpha z(u)]^{-1}$ for $c \geq c_0$ for some
	constant $c_0$.
\end{lemma}
\begin{corollary} \label{cor:prob-halving-upper}
	The probability that any of the long-range contacts of a highway node $u$
	improves its distance to the destination $t$ by a factor of $c$ is at most
	proportional to $[(c - 1)^\alpha z(u)]^{-1}$ for $c \geq c_0$ for some
	constant $c_0$.
\end{corollary}
\begin{proof}
	Assume we are at highway node $u$ which is at distance $d > c \text{ }
	\Theta(\sqrt[\alpha]{k \log{n}})$ from the destination $t$ for some
	arbitrary constant $c > 1$.
	From \Crefpart{lem:balls}{lem:balls:always}, we know that there are
	$\Theta(\frac{d^\alpha}{c^\alpha k})$ highway nodes in
	$\mathcal{B}_{d/c}(u)$ w.h.p.
	The maximum distance between $u$ and any node in $\mathcal{B}_{d/c}(u)$ is
	$d + d/c$.
	Using this, and letting $v$ be an arbitrary long-range connection of $u$,
	we obtain:
	\begin{equation*}
		\Pr(v \in \mathcal{B}_{d/c}(u)) \geq \frac{\Theta\pars*{\frac{ d^\alpha}{c^\alpha k}}}{d^\alpha (1 + 1/c)^\alpha z(u)} = \Theta([(c + 1)^\alpha k z(u) ]^{-1})
	\end{equation*}
	%
	The probability of $v$ not being in $\mathcal{B}_{d/c}(u)$ is
	$1 - \Pr(v \in \mathcal{B}_{d/c}(u)) \leq e^{-\Pr(v \in
	\mathcal{B}_{d/c}(u))}$, using the fact that $1 + x \leq e^x$.
	Recalling that each highway node has $\Theta(k)$ long-range contacts,
	the probability that none of them are in $\mathcal{B}_{d/c}(u)$ is at most
	$e^{-k \Pr(v \in \mathcal{B}_{d/c}(u))} = e^{-\Theta([(c + 1)^\alpha
	z(u)]^{-1})}$.
	Since $z(u)$ is at least a constant, $c$ can be picked large enough such
	that the exponent is small, such that the inequality $e^{-x} \leq 1 - x/2$
	holds, and the probability that none of $u$'s long-range contacts are
	in the ball is at most $1 - \Theta([(c + 1)^\alpha
	z(u)]^{-1})$.
	Consequently, the probability that any of $u$'s long-range contacts
	are in the ball is at least $\Theta([(c + 1)^\alpha z(u)]^{-1})$.
	The corollary follows from a very similar analysis by assuming that all the
	highway nodes are as close to $u$ as possible, i.e., at distance $d - d/c$
	from $u$.
	\qed
\end{proof}
\vspace*{-\medskipamount}
\vspace*{-\medskipamount}
\subsubsection{The Existence of Closer long-range Contacts.}
\label{sec:closer-contacts}
While greedily routing on a fixed-growth graph it is important to stay on the
highway for as long as possible, since there may be significant cost associated
with returning to the highway, e.g., $\mathcal{O}(k)$ hops.
Note that there exist variants of our fixed-growth graphs where this time may be
significantly reduced, or where there is always a closer highway node by
construction,\footnote{As long as the destination is a sufficient distance
away.} but we include this result for the most general and loosest case.
\vspace*{-\medskipamount}
\vspace*{-\medskipamount}
\begin{lemma} \label{lem:closer-contacts}
	For $k \in \Omega(\log{n})$, the greedy routing process stays on
	the highway $\mathcal{H}$ long enough to reach within distance
	$\mathcal{O}(\sqrt[\alpha]{k \log{n}})$ of the destination $t$ w.h.p.
\end{lemma}
\begin{corollary} \label{cor:closer-contacts-exp}
	For $k \in \Omega(\log{n})$, the greedy routing process can expect to stay
	on $\mathcal{H}$ long enough to reach within distance
	$\mathcal{O}(\sqrt[\alpha]{k})$ of an arbitrary destination $t$.
\end{corollary}
\begin{proof}
	When \cite{gila2023highway} prove a similar result for their lattice-based
	randomized highway model, they make extensive use of the underlying lattice
	geometry, specifically the fact where they know exactly how many balls of
	some radius and some distance from $u$ are closer to $t$.
	While we do not have this luxury, we can still achieve similar results in
	this more general setting by using a bit more care.
	
	Specifically, let us consider an arbitrary, optimal (i.e., $d(u, t)$-length) path from $u$ to
	$t$.
	Let $c_b$ be the $1 + (b \times c)$-th node on this path.\footnote{In order to
	achieve independence for using the improved normalization constant from
	\Crefpart{lem:z}{lem:z:constant}, we need to consider nodes at least a distance
	$\Omega(\sqrt[\alpha]{k})$ apart. This can be achieved by letting $c_b$
	refer to the $a \sqrt[\alpha]{k} + b \times w$-th node on the path,
	for some sufficiently large constant $a$, without affecting the results.}
	All nodes within radius $b w$ of $c_b$ are closer to $t$ than to $u$.
	And any node within radius $b w$ of $c_b$ which is not within radius $(b -
	1)w$ of $c_{b - 1}$ is a previously unseen, ``fresh'' node.
	Setting our width $w$ to $\Theta(\sqrt[\alpha]{k \log{n}})$ our results from
	\Cref{lem:shells} directly apply, and we expect to see $\Theta(b^{\alpha -
	1} \log{n})$ fresh nodes within radius $b w$ of $c_b$.
	The maximum distance between $u$ and any such node is $1 + 2 bw =
	\mathcal{O}(bw)$.
	Similarly, there are $\Theta(d(u, t)/w)$ such ball centers $b$ along the
	path.
	Let us consider the probability that one particular long-range connection 
	$v$ of $u$ is closer to $t$.
	\begin{align*}
		\Pr(v \text{ is closer to } t) &\geq \sum_{b = 1}^{\Theta(d(u, t)/w)}{\frac{\text{min \# fresh highway nodes}}{z(u) (\text{max distance to fresh node})^\alpha}} \\
		&\geq z(u)^{-1} \sum_{b = 1}^{\Theta(d(u, t)/w)}{\frac{\Theta(b^{\alpha - 1} \log{n})}{\Theta(b w)^\alpha}} \\
		&= \Theta([k z(u)]^{-1}) \sum_{b = 1}^{\Theta(d(u, t)/w)}{\frac{1}{b}} = \Theta\pars*{\frac{1}{k z(u)} \log{\frac{d(u, t)}{w}}}
	\end{align*}
	Assuming that $d(u, t) \in \Omega(\sqrt[\alpha]{k \log{n}}) = \Omega(w)$, we
	see that the probability that $v$ is closer to $t$ is at least proportional
	to $\log(d(u, t))/(k z(u))$.
	Recalling that $u$ has $\Theta(k)$ long-range contacts, the
	probability that none of them are closer to $t$ is at most
	$d(u, t)^{-\Theta(1/z(u))}$.
	This probability is identical to that used in~\cite{gila2023highway}, and
	the rest of the proof follows for values of $k \in \Omega(\log{n})$.
	The corollary follows by assuming with constant probability that there exist
	highway nodes within distance $\mathcal{O}(\sqrt[\alpha]{k})$ of $t$, along
	with a more optimistic choice of $w$.
	Note that the corollary is in expectation rather than w.h.p.
	\qed
\end{proof}

\subsubsection{Contacts Into and Out Of a Ball.}
When upper bounding the diameter of a randomized highway graph with fixed
growth, it is important to know how quickly the highway nodes can be navigated.
In the worst case, all the already-visited highway nodes are as far as possible
from other highway nodes.
For example, we consider the case where a highway node $u$ has no available
long-range constants within some radius $\ell$.
\begin{lemma} \label{lem:contacts-out}
	The probability that a long-range contact $v$ of a highway node $u$ lies
	outside of $\mathcal{B}_\ell(u)$, for $\ell \leq \sqrt[\alpha]{n^\theta}$
	and $\theta < 1$ is at least proportional to $\frac{\log{n}}{k z(u)}$.
\end{lemma}
\begin{corollary} \label{cor:contacts-in}
	The probability that highway node $u$ is the $i$-th long-range contact of
	any highway node outside of $\mathcal{B}_\ell(u)$ is at least proportional
	to $\frac{\log{n}}{k z}$, where $z$ is the weakest normalization constant bound
	of any highway node from \Crefpart{lem:z}{lem:z:always}.
\end{corollary}

\begin{proof}
	Without loss of generality we can assume that $\ell \in
	\omega(\sqrt[\alpha]{k \log{n}})$ since any smaller value of $\ell$ would
	only make the probability larger.
	Doing this, we can lower bound the probability of $v$ being outside of
	$\mathcal{B}_\ell(u)$ as follows.
	\begin{align*}
		\Pr(v \notin \mathcal{B}_\ell(u)) &\geq \sum_{b = \ell/w}^{\Theta(\sqrt[\alpha]{n}/w)}{\frac{\text{min \# highway nodes in } \mathcal{S}_b^{(w)}}{(\text{max distance to node in } \mathcal{S}_b^{(w)})^\alpha z(u)}} \\
		&= z(u)^{-1} \sum_{b = \ell/w}^{\Theta(\sqrt[\alpha]{n}/w)}{\frac{\Theta(b^{\alpha - 1} \log{n})}{(b w)^\alpha}} \\
		&= \frac{1}{\Theta(k z(u))} (\log{\frac{\sqrt[\alpha]{n}}{w}} - \log{\frac{\ell}{w}}) = \frac{1}{\Theta(k z(u))} (\log{\sqrt[\alpha]{n}} - \log{\ell}) \\
		&\geq \frac{1}{\Theta(k z(u))} (\log{\sqrt[\alpha]{n}} - \log{\sqrt[\alpha]{n^\theta}}) = \frac{\frac{1}{\alpha} \log{n}}{\Theta(k z(u))} (1 - \theta) \\
		&= \Theta\pars*{\frac{\log{n}}{k z(u)}}
	\end{align*}

	The corollary follows by a similar summation in the reverse direction,
	only considering one long-range contact per highway node, and assuming the
	worst case scenario for $z$.
	\qed
\end{proof}

\clearpage

\section{More Experiments}\label{sec:more-exp}

In this section we support our assumption that the estimated dimensionality
$\alpha$ is a better indicator of the optimal clustering exponent than
network size.
Our experimental framework is very similar to that of the recent papers
analyzing U.S. road networks as small-world
graphs~\cite{goodrich2022modeling,gila2023highway}.
We use the same data set and methodology to represent the
graphs~\cite{routingkit,geisberger2012exact}.
The full code for our experiments can be found at
\url{https://github.com/ofekih/Fast-Geographic-Routing-in-Fixed-Growth-Graphs}.

\subsubsection{Estimating Dimensionality and Clustering Coefficient.}
In a fixed-growth graph with constant dimensionality $\alpha$,
$|\mathcal{B}_\ell(u)| \in \Theta(\ell^\alpha)$ for any $u$ and reasonable
values of $\ell > 0$.
When $\ell = 0$, however, there is only one node in the ball, so we can correct
our estimate by adding 1 to the ball size, i.e., $|\mathcal{B}_\ell(u)| = 1 +
\Theta(\ell^\alpha)$, or equivalently, $|\mathcal{B}_\ell(u)| - 1 =
\Theta(\ell^\alpha)$.
Removing the asymptotic notation, we have $c_1 \ell^\alpha \leq
|\mathcal{B}_\ell(u)| - 1 \leq c_2 \ell^\alpha$ for some constants $c_1 < c_2$.
For any finite graph, we can use the distance between $c_1$ and $c_2$ as some
measure of how well the graph fits a fixed-growth model with dimensionality
$\alpha$.

For simplicity, let us momentarily imagine a lattice $\mathcal{L}$.
A more accurate estimate for the ball size of a lattice is actually has an
additional factor of $2\alpha$, i.e., $2\alpha \ell^\alpha$, since nodes
that are not on the edge have $2 \alpha$ local contacts.
Furthermore, if the lattice does not have wrap-around edges, a node in a corner of a
2D lattice will only have one fourth the growth rate of a node in the center, or
$2^{-\alpha}$ times the growth rate for general $\alpha$.
Both of these factors that depend on $\alpha$ would be hidden from the
big-$\mathcal{O}$ notation but would affect the constants $c_1$ and $c_2$.
To minimize their influence, we use the \textit{ratio} $c_2 / c_1$ as our measure of how
well the graph fits a fixed-growth model with dimensionality $\alpha$.
It is then just a matter of finding the $\alpha$ that minimizes this ratio over all nodes with balls of all radii.
Since we dealt with rather large graphs, minimizing the ratio was
computationally expensive, so we instead sampled balls of all sizes for
thousands of random nodes, determined the best $\alpha$ for each node, and
used the median of these $\alpha$ values as our estimate for the overall graph
dimensionality.

In order to estimate the optimal clustering coefficient $s$, we tried a range of
different values for each state, hundreds of thousands of randomized greedy
routing experiments for each.
We used a highway constant $k = \log{n}$, but anecdotally reached very similar
results using $k = 1$ and even by using a power law distribution as was done
in~\cite{goodrich2022modeling}.
Running these experiments was very time consuming, and there was often a range
of clustering exponents that achieved very similar performance
(see~\cite{goodrich2022modeling}), so our results may contain errors of up to
$\pm 0.05$, which may explain some of the variance but should not significantly
affect our conclusions.

\subsubsection{Dimensionality vs. Network Size.}
In a previous analysis of U.S. road networks~\cite{goodrich2022modeling}, it was
conjectured that the optimal clustering exponent should increase as the size of
the state increases, tending towards 2.
From our method of estimating dimensionality from the previous section, it is
indeed true for lattices that as the side length increases, the estimated
dimensionality increases until it approaches the true dimensionality.
The question remains, however, whether it is the network size or the estimated
dimensionality which more directly predicts the optimal clustering exponent.
While it is hard to determine a clear causal relationship, it is clear from our
results in \Cref{fig:indicator_comparison} that our estimated dimensionality
$\alpha$ is a much better indicator than the network size.
These results further support and motivate our fixed-growth model and the
theoretical findings in this paper.

\begin{figure}[h!]
	\centering
	\includegraphics[height=.3\textheight, clip]{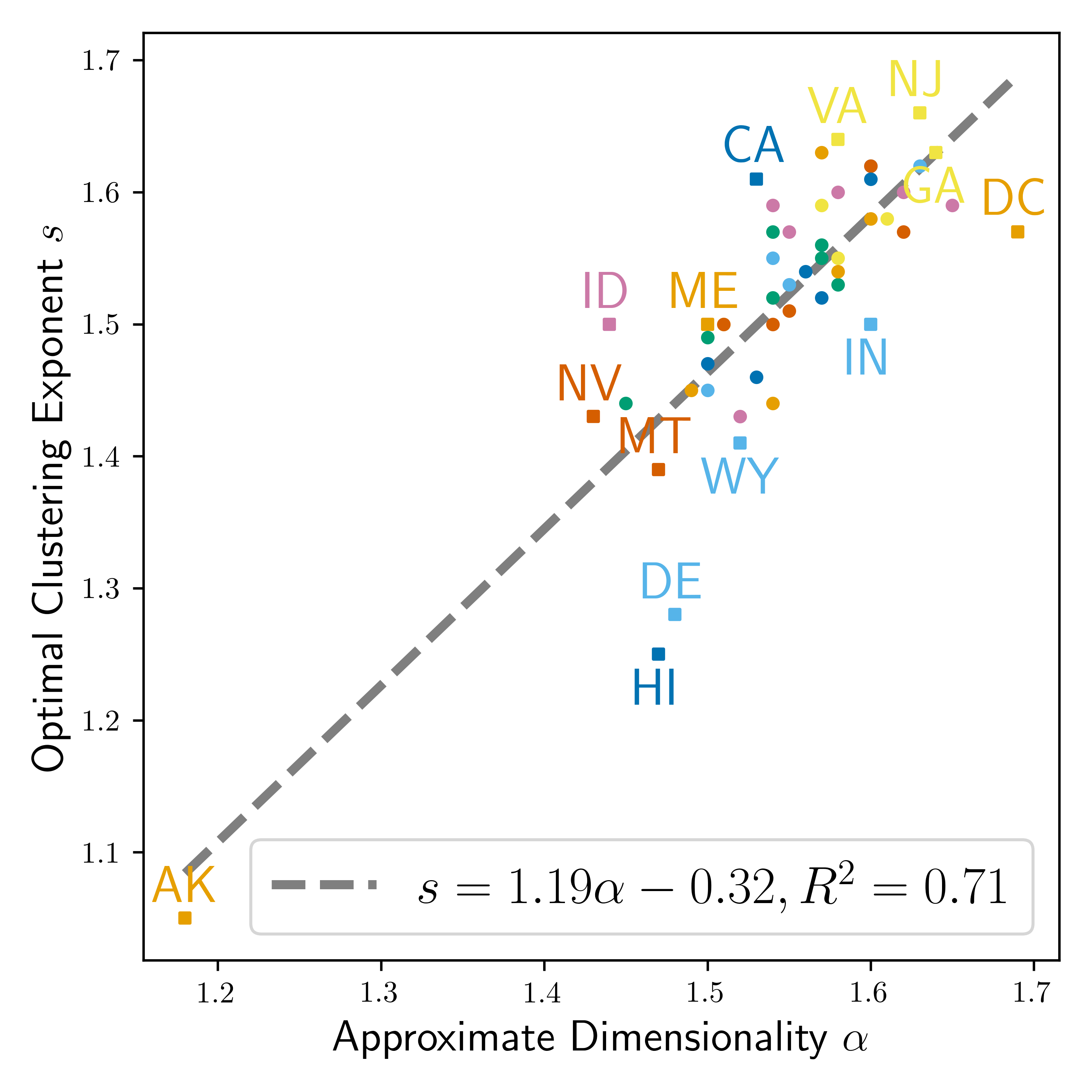}
	\hfill
	\includegraphics[height=.3\textheight,trim={1cm 0 0 0}, clip]{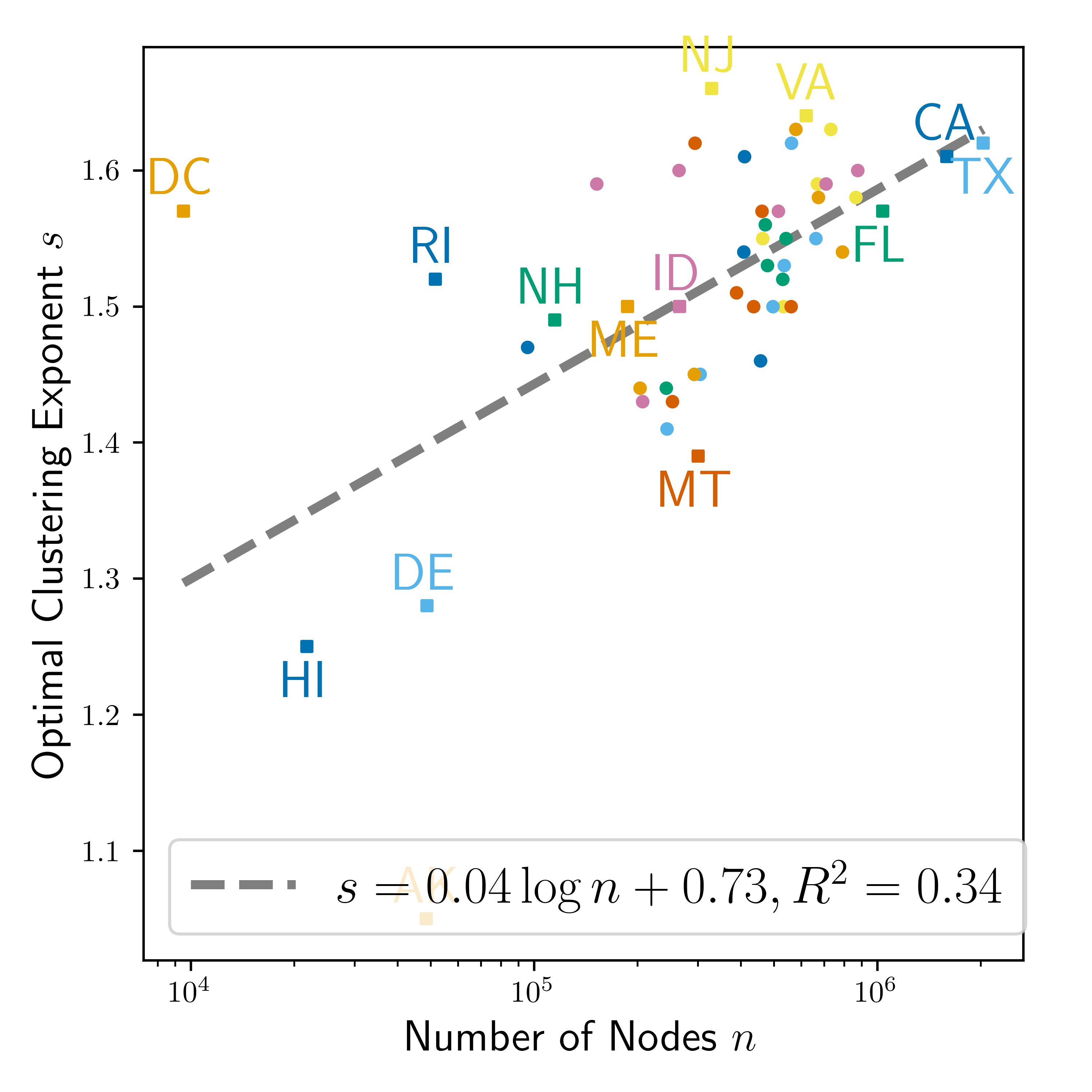}
	\vspace*{-\medskipamount}
	\caption{
		A comparison between the dimensionality $\alpha$ (left) and
		the number of nodes $n$ (right) as indicators for the optimal clustering
		coefficient $s$.
		It is clear that the dimensionality $\alpha$ is a much better indicator
		than the network size.
		\label{fig:indicator_comparison}
	}
\end{figure}

\end{document}